\documentclass[conference]{IEEEtran}
\IEEEoverridecommandlockouts

\usepackage{cite}
\usepackage{amsmath,amssymb,amsfonts}
\usepackage{graphicx}
\usepackage{textcomp}
\usepackage{amsthm}
\usepackage{algorithm}
\usepackage{algpseudocode}
\usepackage{caption}
\usepackage{url}


\usepackage[
  left=0.65in,
  right=0.65in,
  top=0.75in,
  bottom=1.1in
]{geometry}

\usepackage{graphicx}
\usepackage{textcomp}
\usepackage{xcolor}
\newtheorem{Definition}{Definition}
\newtheorem{Lemma}{Lemma}
\newtheorem{Theorem}{Theorem}

\usepackage{comment}

\def\BibTeX{{\rm B\kern-.05em{\sc i\kern-.025em b}\kern-.08em
    T\kern-.1667em\lower.7ex\hbox{E}\kern-.125emX}}

\begin{document}

\title{Insdel Decoding Algorithms for Maximum Distance Separable Codes}

\author{\IEEEauthorblockN{Bencheng Fan}
\IEEEauthorblockA{\textit{Dept. EE Engineering} \\
\textit{The University of Melbourne}\\
Melbourne, Australia \\
benchengf@student.unimelb.edu.au}

\and
\IEEEauthorblockN{Margreta Kuijper}
\IEEEauthorblockA{\textit{Dept. EE Engineering} \\
\textit{The University of Melbourne}\\
Melbourne, Australia \\
mkuijper@unimelb.edu.au}
\and
\IEEEauthorblockN{Jingge Zhu}
\IEEEauthorblockA{\textit{Dept. EE Engineering} \\
\textit{The University of Melbourne}\\
Melbourne, Australia \\
jingge.zhu@unimelb.edu.au}    

}

\maketitle

\begin{abstract}
 In this paper, we focus on decoding MDS linear block codes for insdel errors. Building on our recent work for two-dimensional  Reed-Solomon codes, we present a decoding algorithm for bounded distance decoding for the general insertion-and-deletion scenario that can be used for any MDS code of rate at most half. Our algorithm employs a sliding window and has a time complexity smaller than brute-force decoding. In particular, for optimal Reed-Solomon codes, it corrects insdel errors up to the half-Singleton bound, which exceeds the scope of existing methods.  
 Additionally, for the special cases of insertion-only and deletion-only insdel error patterns, we present two different algorithms with a time complexity smaller than our general algorithm. 
\end{abstract}


\section{Introduction}
DNA storage is an emerging data storage technology that uses synthetic DNA molecules to encode and store digital information. In a DNA-based storage system, information is encoded into sequences over the DNA alphabet $\{A, G, C, T \}$. Because of its high data density, energy efficiency, and longevity, DNA-based storage is a promising approach for meeting the ever-growing demands of data storage \cite{ieee_dna_storage_2024}\cite{9504568}.  

Despite these advantages, DNA-based storage faces several technological challenges, including high costs, slow read and write speeds, and high error rates in the synthesis and sequencing processes. These errors stem from mutations, degradation, and flawed synthesis and sequencing processes. In particular, DNA molecules are susceptible to insertions and deletions (insdel errors)\cite{heckel2019characterization}\cite{8937735}\cite{shomorony2022information}.

Levenshtein, Varshamov and Tenengolts pioneered the study of codes that correct insdel errors in the 1960s \cite{levenshtein1967asymptotically}\cite{levenshtein1966binary}. In \cite{levenshtein1966binary}, Levenshtein demonstrated the equivalence between insertions and deletions by showing that any code capable of correcting $t$ deletions can correct $t$ insertions and deletions.

 It is demonstrated in \cite{cheng2023efficient} that the number of insdel errors that an $[n,k]$ linear code can correct is at most $n-2k+1$, which is called the half-Singleton bound. Linear codes that reach the half-Singleton bound are said to be optimal. Recent literature has shown that it is possible to construct optimal Reed-Solomon codes. \cite{chen2022improved} \cite{con2023reed} \cite{con2024optimal}

So far, the problem of constructing optimal   $[n,k]$ Reed-Solomon codes, where $\frac{n}{2}\geq k\geq 2$ has been addressed. Moreover, designing efficient decoding algorithms for Reed-Solomon codes against insdel errors has also gained attention in recent years. While most work so far has focused on decoding $[n,2]$ Reed-Solomon codes against insdel errors\cite{singhvi2026optimally}\cite{fanKZisita2026}, there are relatively few works addressing the problem of decoding $[n,k]$ Reed-Solomon codes, where $k\geq 3$. For $[n,k]$ optimal Reed-Solomon codes in general where $k\geq 2$, Banerjee \textit{et al.}\cite{banerjee2025decoding} propose a polynomial-time decoding algorithm that can correct up to $t$ insdel errors as long as $t\cdot k =\mathcal{O}(n)$. We observe that a necessary condition for the method proposed in \cite{banerjee2025decoding} to correct $t$ insdel errors is that $t<(k+1)n-\sqrt{(k^2+2k)n^2+kn}$, which is significantly less than the half-Singleton bound $n-2k+1$. \cite{fanKZisita2026}
\subsection{Our Results}
In this paper, we generalise the decoding methods from our recent paper\cite{fanKZisita2026} from $k=2$ to general values of $k$. More specifically, we generalise the results of\cite{fanKZisita2026} from 
optimal $[n,2]$ Reed-Solomon codes to $[n,k]$ MDS codes (Maximum Distance Separable codes) that are not necessarily optimal. Such codes correct any erasure pattern with $n-k$ erasures. Throughout the paper, we assume that $k\leq\frac{n}{2}$ is a fixed value independent of $n$. For $[n,k]$ MDS codes, our results are threefold. Firstly, we present the main result, which is a novel algorithm correcting insdel errors up to the error-correcting capability of the code with a time complexity of $\mathcal{O}(n^{k+3})$. For the special case of optimal $[n,k]$ Reed-Solomon codes, our method corrects up to the half-Singleton bound, which exceeds the scope of \cite{banerjee2025decoding}. We then propose a decoding algorithm that corrects deletions up to the error-correcting capability of the code with a time complexity of $\mathcal{O}(n^{k+1})$. Lastly, we propose a decoding algorithm that corrects insertions up to the error-correcting capability of the code with a time complexity of $\mathcal{O}(n^{k+1})$. Whereas the two algorithms corresponding to the deletions-only and insertions-only scenarios are straightforward generalisations of the methods from\cite{fanKZisita2026}, the decoding algorithm for MDS codes in general, where both deletions and insertions may happen, employs a new technique, namely the sliding window, and is a less straightforward generalisation of the method from\cite{fanKZisita2026}. 

\subsection{Organization}
The rest of the paper is organised as follows. In Section II, we introduce some key definitions and known results. In Section III, we present the three decoding algorithms for MDS codes and apply the first one to optimal $[n,k]$ Reed-Solomon codes up to the half-Singleton bound. In Section IV, we briefly summarise the results presented in this paper.

\section{Preliminaries}
For a sequence $\mathbf{x}$ with a length of at least $m$ and a vector of indices $\mathbf{J}=(j_1,j_2,\cdots,j_m)$, we denote the subsequence $(x_{j_1},x_{j_2},\cdots,x_{j_m})$ as $\mathbf{x_J}$. Alternatively, we sometimes also denote it as $\mathbf{x}(j_1,j_2,\cdots,j_m)$.  Moreover, if $\mathbf{J}=(1,2,\cdots,m)$, we denote $\mathbf{x_J}$ as $\mathbf{x}(1:m)$. Throughout, we denote $\mathbb{F}$ as the finite field. For a matrix $G\in\mathbb{F}^{k\times n}$, we denote its column vectors as $\mathbf{g}_i$, where $i\in\{1,2,\cdots,n\}$. For a vector of indices $\mathbf{J}$, we denote the $k\times m$ matrix $[\mathbf{g}_{j_1},\mathbf{g}_{j_2},\cdots,\mathbf{g}_{j_m}]$ as $G_J$.

\begin{Definition}
    Let $\mathbf{x}$ be a sequence. If a sequence $\mathbf{y}$ can be obtained by deleting certain symbols from $\mathbf{x}$, then $\mathbf{y}$ is called a subsequence of $\mathbf{x}$. If a sequence $\mathbf{y}$ is a subsequence of both $\mathbf{x_1}$ and $\mathbf{x_2}$, then $\mathbf{y}$ is called a common subsequence of $\mathbf{x_1}$ and $\mathbf{x_2}$.
\end{Definition}

\begin{Definition}
    For two sequences $\mathbf{x},\mathbf{y}$, the insdel distance $d(\mathbf{x},\mathbf{y})$ is defined as the minimum number of insertions and deletions to transform $\mathbf{x}$ into $\mathbf{y}$.
\end{Definition}

Similar to the minimum Hamming distance in traditional error-correction, we define the minimum insdel distance of a code.
\begin{Definition}
    For a code $\mathcal{C}$, its minimum insdel distance $d$ is defined as

\begin{equation}
  d=\min_{\mathbf{c_1},\mathbf{c_2}\in C,\mathbf{c_1}\neq \mathbf{c_2}}\{d(\mathbf{c_1},\mathbf{c_2})\}  
\end{equation}
\end{Definition}
It is demonstrated in \cite{cheng2023efficient} that there exists a tight upper bound on the minimum insdel distance of an $[n,k]$ linear code.
\begin{Theorem}
\label{H-S}
(\cite{cheng2023efficient})
    For an $[n,k]$ code, its minimum insdel distance $d$ must satisfy $d\leq2n-4k+4$. 
\end{Theorem}
An $[n,k]$ code with minimum insdel distance $d=2n-4k+4$ is said to be optimal.
In \cite{levenshtein1966binary}, Levenshtein demonstrated that the insdel distance is a metric. He also showed that a code's capability to correct insdel errors is directly decided by its minimum insdel distance.
\begin{Lemma}
\label{Insdel}
(\cite{levenshtein1966binary}) Let $\mathcal{C}$ be a code over $\mathbb{F}$ with minimum insdel distance $d$. The code $\mathcal{C}$ can correct $t$ insdel errors if and only if $2t<d$.
\end{Lemma}

Calculating the minimum insdel distance is essential for assessing a code's ability to correct insdel errors. It is well-known (see e.g.\cite[Lemma 1]{singhvi2026optimally}) that the insdel distance between two sequences can be calculated via their longest common
subsequence. We recall this result in the following lemma.
\begin{Lemma}
\label{LCS}
     Let $\mathbf{a}\in\mathbb{F}^{n_1}$ and $\mathbf{b}\in\mathbb{F}^{n_2}$, then we have
\begin{equation}
    d(\mathbf{a},\mathbf{b})=n_1+n_2-2LCS(\mathbf{a},\mathbf{b}),
\end{equation}
where $LCS(\mathbf{a},\mathbf{b})$ denotes the length of the longest common subsequence of $\mathbf{a}$ and $\mathbf{b}$. 
\end{Lemma}
Lemma \ref{LCS} translates the problem of computing the insdel distance between two sequences to determining the length of their longest common subsequence. In \cite{hirschberg1977algorithms}, an algorithm that solved the longest common subsequence problem in $\mathcal{O}(n^2)$ time was presented. 

As a benchmark for comparison, we first present the brute-force method for decoding $[n,k]$ MDS codes against up to $\lfloor\frac{d-1}{2}\rfloor$ insdel errors and then give an analysis of its time complexity.

\begin{algorithm}[!b]
\caption{Brute-Force Decoding Method}
\label{Brute}

\begin{algorithmic}[1]
\Require Received sequence $\mathbf{y}$ of length $\ell$, generator matrix $G\in\mathbb{F}^{k\times n}$,  minimum insdel distance $d$
\Ensure Decoded codeword $\mathbf{c}_o$ or \textsc{false}

\If{$\ell \leq n-\frac{d}{2}$ or $\ell \geq n+\frac{d}{2}$}
    \State \Return \textsc{false}
\EndIf

\For{all $(i_1,\cdots,i_k)$ such that $1 \leq i_1 <\cdots<i_k \leq l$}
    \State $\mathbf{I} \gets (i_1,\cdots,i_k)$
    \For{all $(j_1,\cdots,j_k)$ such that $1 \leq j_1<\cdots<j_k \leq n$}
        \State $\mathbf{J} \gets (j_1,\cdots,j_k)$
        \State Solve $\mathbf{u}\cdot G_{\mathbf{J}}=\mathbf{y}_{\mathbf{I}}$
        \State $\mathbf{c} \gets \mathbf{u}\cdot G$

        \If{$d(\mathbf{c},\mathbf{y})< \frac{d}{2}$}
            \State $\mathbf{c}_o \gets \mathbf{c}$
            \State \Return $\mathbf{c}_o$
        \EndIf
    \EndFor
    
\EndFor

\State \Return \textsc{false}
\end{algorithmic}
\end{algorithm}

\begin{Theorem}
\label{Theo: Brute}
    Let $\mathcal{C}$ be an $[n,k]$ MDS code with minimum insdel distance $d$. Algorithm \ref{Brute} correctly decodes up to $\lfloor \frac{d-1}{2}\rfloor$ insdel errors in  $\mathcal{O}(n^{2k+2})$ time.  
\end{Theorem}
\begin{proof}
    Let $\mathbf{c_o}$ be the true codeword. As we assume that the number of insdel errors is at most $\lfloor \frac{d-1}{2}\rfloor$, we have that $n-\frac{d}{2}<\ell < n+\frac{d}{2}$. Furthermore, since the number of deletions is smaller than $\frac{d}{2}$, by Theorem \ref{H-S}, there are at least $n-(n-2k+1)=2k-1\geq k$ symbols from $\mathbf{c_o}$ in $\mathbf{y}$. Let $\mathbf{I}=(i_1,\cdots,i_k)$ and $\mathbf{J}=(j_1\cdots,j_k)$.
    Because the code $\mathcal{C}$ is MDS, the matrix $\mathbf{G_J}$ is full rank, and there exists a unique solution to the linear equation $\mathbf{u}\cdot\mathbf{G_J}=\mathbf{y_I}$, which we denote as $\mathbf{u}$. Suppose that $\mathbf{c}=\mathbf{u}\cdot G$ satisfies $d(\mathbf{c},\mathbf{y})< \frac{d}{2}$. We now show by contradiction that $\mathbf{c}$ must be equal to $\mathbf{c_o}$. Suppose that $\mathbf{c}\neq \mathbf{c_o}$. 
    By assumption, $d(\mathbf{c_o},\mathbf{y})< \frac{d}{2}$, and this implies that $d(\mathbf{c,\mathbf{c_o}})\leq d(\mathbf{c},\mathbf{y})+d(\mathbf{c_o},\mathbf{y})<d$, which contradicts the fact that $d$ is the minimum insdel distance of the code $\mathcal{C}$. It follows that $\mathbf{c}= \mathbf{c_o}$. If $d(\mathbf{c},\mathbf{y})\geq\frac{d}{2}$, the algorithm moves on to the next combination of $(\mathbf{I},\mathbf{J})$. If no codeword satisfying $d(\mathbf{c},\mathbf{y})< \frac{d}{2}$ is found for all combinations of $(\mathbf{I},\mathbf{J})$, then it indicates that the number of insdel errors exceeds the insdel error-correction capability of the code. The algorithm returns FALSE. The time complexity of
    Algorithm \ref{Brute} is calculated as follows.  For each combination of $(\mathbf{I},\mathbf{J})$, the encoding process in step 9 takes $\mathcal{O}(n)$ time. It takes further $\mathcal{O}(n^2)$ time to calculate $d(\mathbf{c},\mathbf{y})$ \cite{hirschberg1977algorithms}. Since there are $\binom{\ell}{k}$ and $\binom{n}{k}$ combinations for $\mathbf{I}$ and $\mathbf{J}$, respectively, the time complexity of the brute-force method can be calculated as $\big(\mathcal{O}(n)+\mathcal{O}(n^2)\big)\cdot \mathcal{O}(n^k)\cdot \mathcal{O}(n^k)=\mathcal{O}(n^{2k+2})$. 
\end{proof}

\section{Main Results}

In this section, we first present an insdel decoding algorithm for $[n,k]$ MDS codes. Then, we derive decoding algorithms for optimal $[n,k]$ Reed-Solomon codes for decoding insdel errors up to the half-Singleton bound, \textit{i.e.}, the number of insdel errors must be less than or equal to $n-2k+1$. We denote the number of insertions and deletions as $r$ and $s$, respectively.  Note that, upon receiving the sequence $\mathbf{y}$, we do not know the values of $r$ and $s$, but we do know the length $\ell$ of $\mathbf{y}$. 
\subsection{Decoding MDS Codes Against Insdel Errors}
We now present the main result of the paper, which is an algorithm for decoding $[n,k]$ MDS codes for up to $\lfloor\frac{d-1}{2}\rfloor$ insdel errors, allowing for both deletions and insertions. Our algorithm generalises the method for $k=2$ presented in \cite{fanKZisita2026} to MDS codes with $k\geq 2$ by introducing a sliding window decoder.

\begin{algorithm}[!b]
\caption{Algorithm for Decoding MDS Codes Against  Insdel Errors}
\label{MDS}
\begin{algorithmic}[1]
\Require Received sequence $\mathbf{y}$ of length $\ell$, generator matrix $G\in\mathbb{F}^{k\times n}$,  minimum insdel distance $d$
\Ensure Decoded codeword $\mathbf{c}_o$ or \textsc{false}

\If{$\ell \leq n-\frac{d}{2}$ or $\ell \geq n+\frac{d}{2}$}
    \State \Return \textsc{false}
\EndIf

\For{$t\in\{1,2,\cdots,\ell-2k+3\}$}
    \State $\mathbf{T}\gets(t,t+1,\cdots,t+2k-3)$
    \For{$\{i_1,i_2,\cdots,i_k\}\subseteq T$ and $i_1<i_2<\cdots<i_k$}
        \State $\mathbf{I}\gets(i_1,i_2,\cdots,i_k)$
        \For{$1\leq j_1<j_2<\cdots<j_k\leq n$}
            \State $\mathbf{J}\gets(j_1,j_2,\cdots,j_k)$
            \State Solve $\mathbf{u}\cdot G_{\mathbf{J}}=\mathbf{y}_{\mathbf{I}}$
            \State $\mathbf{c}\gets\mathbf{u}\cdot G$
            \If{$d(\mathbf{c},\mathbf{y})< \frac{d}{2}$}
                \State $\mathbf{c_o}\gets\mathbf{c}$
                \State \Return $\mathbf{c_o}$
            \EndIf
        \EndFor
    \EndFor
\EndFor

\State \Return \textsc{false}
\end{algorithmic}
\end{algorithm}

To facilitate the proof of Algorithm \ref{MDS}'s correctness, we introduce the following Lemmas.

\begin{Lemma}\label{Seq_X}
\label{series}
     Consider a binary sequence whose terms satisfy
\[
x_n \in \{0,1\}, \qquad n \geq 1.
\] For $p\geq 0$, define its corresponding sum sequence $S_p$ as 
     $$S_p=\begin{cases}
    \sum_{n=1}^px_n & \text{if } p > 0, \\
    0   & \text{if } p = 0.
\end{cases}$$ Let $m$ be a positive integer. If:
     $$ S_{n+2m}-S_{n}\leq m\;\;\;\;\;\;\text{for all } n\geq 0,$$
     then for all $p\geq 0$:
$$S_p\leq \frac{1}{2}(p+m).$$

\end{Lemma}

\begin{proof}
    First, we examine the scenarios where $0\leq p\leq m$. As $p \leq m$, we have that $S_p \leq p=\frac{1}{2}(p+p)\leq \frac{1}{2}(p+m)$. Second, assume that $m+1\leq p\leq 2m$, we have $S_p\leq S_{2m}=S_{0+2m}-S_0\leq m=\frac{1}{2}(m+m)<\frac{1}{2}(p +m)$.\\
    \indent We now use induction to prove that $S_p\leq\frac{1}{2}(p+m)$ holds for all $p>2m$. Assume that there exists a positive integer $q\geq 2m$ such that $S_i\leq \frac{1}{2}(i+m)$ for all $0\leq i \leq q$. Since $S_{q+1}-S_{q+1-2m}\leq m$, we have that $S_{q+1}\leq S_{q+1-2m}+ m\leq \frac{1}{2}(q+1-2m+m)+m=\frac{1}{2}(q+1+m)$. In other words, $S_i\leq \frac{1}{2}(i+m)$ for all $0\leq i \leq q+1$. As is previously shown, $S_i\leq \frac{1}{2}(i+m)$ for all $0\leq i \leq 2m$. By induction, we conclude that $S_p\leq \frac{1}{2}(p+m),\text{for all } p\geq 0$.
\end{proof}

Based on Lemma \ref{Seq_X}, we derive the following Lemma, which is the foundation for the key idea of Algorithm \ref{MDS}.

\begin{Lemma}
\label{adjacent}
    Let $\mathbf{c}\in \mathbb{F}^n$ be a sequence sent by the transmitter. Let $\mathbf{y}$ be the received sequence of length $\ell$ obtained by deleting $s$ symbols from and inserting $r$ symbols into $\mathbf{c}$, where $r+s< \frac{d}{2}$. Then, there exists  $1\leq i\leq \ell-2k+3$ such that the subsequence $\mathbf{y}(i,i+1,\cdots,i+2k-3)$ contains at least $k$ symbols from $\mathbf{c}$. 
\end{Lemma}
\begin{proof}
     Let $p$ be a positive integer. We denote the number of symbols from $\mathbf{c}$ in $\mathbf{y}(1:p)$ as $S_p$. Our proof is by contradiction as follows: assume that such an $i$ does not exist; then for all $i\in\{0,\cdots,\ell-2k+2\}$, we have that $S_{i+2k-2}-S_i\leq k-1$. Then, by  Lemma \ref{series}, it follows that 
     \begin{equation}
        S_\ell\leq\frac{1}{2}(\ell+k-1).   
     \end{equation}
      On the other hand,  $r+s< \frac{d}{2}\leq n-2k+2< n-k+1$, which indicates that $n-s>r+k-1$. Therefore, we have that $S_\ell=n-s>\frac{1}{2}(n-s+r+k-1)=\frac{1}{2}(\ell+k-1)$. This contradicts (3) and therefore proves the Lemma.  
\end{proof}
 
We now provide a proof for Algorithm \ref{MDS} and give an analysis of its time complexity.  
\begin{Theorem}
\label{Theo: General}
     Let $\mathcal{C}$ be an $[n,k]$ MDS code with minimum insdel distance $d$. Algorithm \ref{MDS} correctly decodes up to $\lfloor \frac{d-1}{2}\rfloor$ insdel errors in  $\mathcal{O}(n^{k+3})$ time. 
\end{Theorem}

\begin{proof}
    Let $\mathbf{c_o}$ be the true codeword. 
    By Lemma \ref{adjacent}, there exists $1\leq t\leq \ell-2k+3$ such that the subsequence $\mathbf{y}(t,t+1,\cdots,t+2k-3)$ contains at least $k$ symbols from $\mathbf{c_o}$. For each possible value of $t$, denote $T=\{t,t+1,\cdots,t+2k-3\}$. Let $\mathbf{I}=(i_1,i_2,\cdots,i_k)$, where $\{i_1,i_2,\cdots,i_k\}\subseteq T$ and $i_1<i_2<\cdots<i_k$. Let $\mathbf{J}=(j_1,j_2,\cdots,j_k)$, where $1\leq j_1<j_2<\cdots<j_k\leq n$. Denote the unique solution of the linear equation $\mathbf{u}\cdot\mathbf{G_J}=\mathbf{y_I}$ as $\mathbf{u}$. Suppose that $\mathbf{c}=\mathbf{u}\cdot G$ satisfies $d(\mathbf{c},\mathbf{y})< \frac{d}{2}$. We now show by contradiction that $\mathbf{c}$ must be equal to $\mathbf{c_o}$. Suppose that $\mathbf{c}\neq \mathbf{c_o}$. 
    By assumption, $d(\mathbf{c_o},\mathbf{y})< \frac{d}{2}$, and this implies that $d(\mathbf{c,\mathbf{c_o}})\leq d(\mathbf{c},\mathbf{y})+d(\mathbf{c_o},\mathbf{y})<d$, which contradicts the fact that $d$ is the minimum insdel distance of the code $\mathcal{C}$. It follows that $\mathbf{c}= \mathbf{c_o}$. If $d(\mathbf{c},\mathbf{y})\geq \frac{d}{2}$, the algorithm moves on to the next combination of $(t,\mathbf{I},\mathbf{J})$. If no codeword satisfying $d(\mathbf{c},\mathbf{y})< \frac{d}{2}$ is found for all combinations of $(t,\mathbf{I},\mathbf{J})$, then it indicates that the number of insdel errors exceeds the error correction capability of the code. The algorithm returns FALSE. 
    The time complexity of Algorithm \ref{MDS} is calculated as follows. As $t$ takes values from the set $t\in\{1,2,\cdots,\ell-2k+3\}$, a time complexity of $\mathcal{O}(n)$ is introduced.
    For each value of $t$, there are exactly $\binom{2k-2}{k}$ possible combinations for $\mathbf{I}$ and $\binom{n}{k}$ combinations for $\mathbf{J}$. As we assume that $k$ is a fixed integer, the $\binom{2k-2}{k}$ combinations for $\mathbf{I}$ introduce a time complexity of $\mathcal{O}(1)$. The $\binom{n}{k}$ combinations for $\mathbf{J}$, on the other hand, introduce a time complexity of $\mathcal{O}(n^k)$. For each combination of $(t,\mathbf{I},\mathbf{J})$, it takes $\mathcal{O}(n)$ time for the encoding process in step 11. It takes further $\mathcal{O}(n^2)$ time to calculate $d(\mathbf{c},\mathbf{y})$\cite{hirschberg1977algorithms}.
    Therefore, the total time complexity is $ \mathcal{O}(n)\cdot\mathcal{O}(1)\cdot\mathcal{O}(n^k)\cdot\big(\mathcal{O}(n)+\mathcal{O}(n^2)\big)=\mathcal{O}(n^{k+3})$. 
\end{proof}

To establish connections with the recent literature on optimal Reed-Solomon codes, we now present Algorithm \ref{General} as a special case of Algorithm \ref{MDS}, where the MDS code is optimal, \textit{i.e.} $d=2n-4k+4$.
By Theorem \ref{Theo: General}, our method corrects up to the half-Singleton bound, which exceeds the scope of \cite{banerjee2025decoding}.

\begin{algorithm}[!b]
\caption{Algorithm for Optimal MDS Codes}
\label{General}
\begin{algorithmic}[1]
\Require Received sequence $\mathbf{y}$ of length $\ell$, generator matrix $G\in\mathbb{F}^{k\times n}$
\Ensure Decoded codeword $\mathbf{c}_o$ or \textsc{false}

\If{$\ell\leq2k-2$ or $\ell\geq2n-2k+2$}
    \State \Return \textsc{false}
\EndIf

\For{$t\in\{1,2,\cdots,\ell-2k+3\}$}
    \State $\mathbf{T}\gets(t,t+1,\cdots,t+2k-3)$
    \For{$\{i_1,i_2,\cdots,i_k\}\subseteq T$ and $i_1<i_2<\cdots<i_k$}
        \State $\mathbf{I}\gets(i_1,i_2,\cdots,i_k)$
        \For{$1\leq j_1<j_2<\cdots<j_k\leq n$}
            \State $\mathbf{J}\gets(j_1,j_2,\cdots,j_k)$
            \State Solve $\mathbf{u}\cdot G_{\mathbf{J}}=\mathbf{y}_{\mathbf{I}}$
            \State $\mathbf{c}\gets\mathbf{u}\cdot G$
            \If{$d(\mathbf{c},\mathbf{y})\le n-2k+1$}
                \State $\mathbf{c_o}\gets\mathbf{c}$
                \State \Return $\mathbf{c_o}$
            \EndIf
        \EndFor
    \EndFor
\EndFor

\State \Return \textsc{false}
\end{algorithmic}
\end{algorithm}

\subsection{Decoding MDS Codes for Deletions}
For the deletion-only scenario, the received sequence $\mathbf{y}$ is a subsequence of the transmitted sequence $\mathbf{c_o}$. The algorithm we present below is a straightforward generalisation of the method from our recent paper \cite{fanKZisita2026}.
\begin{algorithm}[!b]
\caption{Algorithm for the Deletion-Only Scenario}
\label{Deletion}

\begin{algorithmic}[1]
\Require Received sequence $\mathbf{y}$ of length $\ell$, generator matrix $G\in\mathbb{F}^{k\times n}$, minimum insdel distance $d$
\Ensure Decoded codeword $\mathbf{c}_o$ or \textsc{false}

\If{$\ell \leq n-\frac{d}{2}$}
    \State \Return \textsc{false}
\EndIf

\State $\mathbf{I} \gets (1,2,\cdots,k)$

\For{all $(j_1,j_2,\cdots,j_k)$ such that $1 \leq j_1 < j_2<\cdots<j_k \leq k+\lfloor \frac{d-1}{2}\rfloor$}
    \State $\mathbf{J} \gets (j_1,j_2,\cdots,j_k)$
    \State Solve $\mathbf{u}\cdot G_{\mathbf{J}}=\mathbf{y}_{\mathbf{I}}$
    \State $\mathbf{c} \gets \mathbf{u}\cdot G$

    \If{$\mathbf{y}$ is a subsequence of $\mathbf{c}$}
        \State $\mathbf{c}_o \gets \mathbf{c}$
        \State \Return $\mathbf{c}_o$
    \EndIf
\EndFor

\State \Return \textsc{false}
\end{algorithmic}
\end{algorithm}

\begin{Theorem}
\label{Theo: Deletion}
    Let $\mathcal{C}$ be an $[n,k]$ MDS code with minimum insdel distance $d$. Algorithm \ref{Deletion} correctly decodes up to $\lfloor \frac{d-1}{2}\rfloor$ deletions in  $\mathcal{O}(n^{k+1})$ time. If the number of deletions exceeds $\lfloor \frac{d-1}{2}\rfloor$, then Algorithm \ref{Deletion} returns FALSE.
\end{Theorem}

\begin{proof}
    Let $\mathbf{c_o}$ be the true codeword. In a
    deletion-only scenario, $\mathbf{y}(1:k)$ must be a subsequence of $\mathbf{c_o}$, say $\mathbf{c_o}(j_1,j_2,\cdots,j_k)$. As we assume that the number of deletions is at most $\lfloor \frac{d-1}{2}\rfloor$, we have that $j_k\leq k+\lfloor \frac{d-1}{2}\rfloor$. Therefore, $j_1, j_2,\cdots, j_k$ satisfy $1 \leq j_1 < j_2<\cdots<j_k \leq k+\lfloor \frac{d-1}{2}\rfloor$. Let $\mathbf{I}=(1,2,\cdots,k)$ and $\mathbf{J}=(j_1,j_2,\cdots,j_k)$. Denote the unique solution of the linear equation $\mathbf{u}\cdot\mathbf{G_J}=\mathbf{y_I}$ by $\mathbf{u}$. Suppose that $\mathbf{y}$ is a subsequence of $\mathbf{c}=\mathbf{u}\cdot G$. We now show that $\mathbf{c}$ must be equal to $\mathbf{c_o}$. Since $\mathbf{y}$ is a common subsequence of $\mathbf{c}$ and $\mathbf{c_o}$, it follows from Lemma \ref{LCS} that $d(\mathbf{c},\mathbf{c_o})\leq2n-2\ell<2n-2(n-\frac{d}{2})=d$, which implies that $\mathbf{c}=\mathbf{c_o}$.
    The time complexity of Algorithm \ref{Deletion} is calculated as follows. For each $(j_1,j_2,\cdots,j_k)$, the encoding process in step 8 takes $\mathcal{O}(n)$ time. To determine if $\mathbf{y}$ is a subsequence of $\mathbf{c}$, we utilize a standard greedy two-pointer approach, which runs in linear time $\mathcal{O}(n)$\cite{laaksonen2020guide}. Since there are $\binom{k+\lfloor \frac{d-1}{2}\rfloor}{k}$ combinations of $(j_1,j_2,\cdots,j_k)$, the total time complexity is $\mathcal{O}(n^k)\cdot\big(\mathcal{O}(n)+\mathcal{O}(n)\big)=\mathcal{O}(n^{k+1})$. If all possible combinations of $(j_1,j_2,\cdots,j_k)$ fail, then either the number of insertions is not zero or the number of deletions exceeds $\lfloor \frac{d-1}{2}\rfloor$. The algorithm returns FALSE.
\end{proof}

\subsection{Decoding MDS Codes for Insertions}
For the insertion-only scenario, the transmitted sequence $\mathbf{c_o}$ is a subsequence of the received sequence $\mathbf{y}$. The algorithm we present below is a straightforward generalisation of the method from our recent paper \cite{fanKZisita2026}.

\begin{algorithm}[!b]
\caption{Algorithm for the Insertion-Only Scenario}
\label{Insertion}

\begin{algorithmic}[1]
\Require Received sequence $\mathbf{y}$ of length $\ell$, generator matrix $G\in\mathbb{F}^{k\times n}$, minimum insdel distance $d$
\Ensure Decoded codeword $\mathbf{c}_o$ or \textsc{false}

\If{$\ell \geq n+\frac{d}{2}$}
    \State \Return \textsc{false}
\EndIf

\State $\mathbf{J} \gets (1,2,\cdots,k)$

\For{all $(i_1,i_2,\cdots,i_k)$ such that $1 \leq i_1 < i_2 <\cdots<i_k\leq k+\lfloor \frac{d-1}{2}\rfloor$}
    \State $\mathbf{I} \gets (i_1,i_2,\cdots,i_k)$
    \State Solve $\mathbf{u}\cdot G_{\mathbf{J}}=\mathbf{y}_{\mathbf{I}}$
    \State $\mathbf{c} \gets \mathbf{u}\cdot G$

    \If{$\mathbf{c}$ is a subsequence of $\mathbf{y}$}
        \State $\mathbf{c}_o \gets \mathbf{c}$
        \State \Return $\mathbf{c}_o$
    \EndIf
\EndFor

\State \Return \textsc{false}
\end{algorithmic}
\end{algorithm}

\begin{Theorem}
\label{Theo: Insertion}
    Let $\mathcal{C}$ be an $[n,k]$ MDS code with minimum insdel distance $d$. 
    Algorithm \ref{Insertion} correctly decodes up to $\lfloor \frac{d-1}{2}\rfloor$ insertions in  $\mathcal{O}(n^{k+1})$ time. If the number of insertions exceeds $\lfloor \frac{d-1}{2}\rfloor$, then Algorithm \ref{Insertion} returns FALSE.
\end{Theorem}

\begin{proof}
    Let $\mathbf{c_o}$ be the true codeword. In an insertion-only scenario, $\mathbf{c_o}(1:k)$ must be a subsequence of $\mathbf{y}$, say $\mathbf{y}(i_1,i_2,\cdots,i_k)$.  As we assume that the number of insertions is at most $\lfloor \frac{d-1}{2}\rfloor$, we have that $i_k\leq k+\lfloor \frac{d-1}{2}\rfloor$. Therefore, $i_1,i_2,\cdots,i_k$ satisfy $1 \leq i_1 < i_2 <\cdots<i_k\leq k+\lfloor \frac{d-1}{2}\rfloor$. Let $\mathbf{J}=(1,2,\cdots,k)$ and $\mathbf{I}=(i_1,i_2,\cdots,i_k)$. Denote the unique solution of the linear equation $\mathbf{u}\cdot\mathbf{G_J}=\mathbf{y_I}$ by $\mathbf{u}$.  Suppose $\mathbf{c}=\mathbf{u}\cdot G$ is a subsequence of $\mathbf{y}$. We now show that $\mathbf{c}$ must be equal to $\mathbf{c_o}$. Since both $\mathbf{c}$ and $\mathbf{c_o}$ are subsequences of $\mathbf{y}$, we have $d(\mathbf{c},\mathbf{y})=d(\mathbf{c_o},\mathbf{y})=\ell-n$. Therefore, by the triangle inequality, $d(\mathbf{c},\mathbf{c_o})\leq d(\mathbf{c},\mathbf{y})+d(\mathbf{c_o},\mathbf{y})\leq 2\ell-2n< 2n+d-2n=d$, which implies that $\mathbf{c}= \mathbf{c_o}$. The time complexity of Algorithm \ref{Insertion} is calculated as follows. For each $(i_1,i_2,\cdots,i_k)$, the encoding process in step 8 takes $\mathcal{O}(n)$ time. To determine if $\mathbf{c}$ is a subsequence of $\mathbf{y}$, we utilize a standard greedy two-pointer approach, which runs in linear time $\mathcal{O}(n)$\cite{laaksonen2020guide}. Since there are $\binom{k+\lfloor \frac{d-1}{2}\rfloor}{k}$ combinations of $(i_1,i_2,\cdots,i_k)$, the total time complexity is $\mathcal{O}(n^k)\cdot\big(\mathcal{O}(n)+\mathcal{O}(n)\big)=\mathcal{O}(n^{k+1})$. If all possible combinations of $(i_1,i_2,\cdots i_k)$ fail, then either the number of deletions is not zero or the number of insertions exceeds $\lfloor \frac{d-1}{2}\rfloor$. The algorithm returns FALSE.
\end{proof}

\section{Conclusions}
In this paper, we focus on efficient insdel decoding algorithms for $[n,k]$ MDS codes. Our main result is a bounded distance decoding algorithm for decoding insdel errors in general with a time complexity of $\mathcal{O}(n^{k+3})$. It improves on the brute-force method with a time complexity of $\mathcal{O}(n^{2k+2})$. Additionally, we presented two algorithms with a time complexity of $\mathcal{O}(n^{k+1})$ for the simpler deletion-only and insertion-only scenarios, respectively. Our methods utilise only the MDS property and do not require the codes to be optimal in the insdel setting. A possible avenue for future research is to develop more efficient insdel bounded distance decoding algorithms for $[n,k]$ Reed-Solomon codes exploiting their special algebraic properties.

\bibliographystyle{plain}
\bibliography{references}

\section* {AI Statement}
Artificial intelligence tools have not been used to produce any parts of this manuscript, including but not limited to text and code. The planning, analysis, and writing of this research were executed entirely by the human authors.

\end{document}